\documentclass[aps,pra,reprint,nofootinbib,showpacs,superscriptaddress]{revtex4-1}

\usepackage{amsmath, amssymb, amsthm}
\usepackage{bbm}
\usepackage{dsfont}
\usepackage{mathrsfs}
\usepackage{mathbbol}
\usepackage{stmaryrd}
\usepackage{physics}
\usepackage{enumitem}
\DeclareSymbolFont{largesymbolsA}{U}{txexa}{m}{n}
\DeclareMathSymbol{\varprod}{\mathop}{largesymbolsA}{16}

\DeclareFontFamily{U}{mathx}{\hyphenchar\font45}
\DeclareFontShape{U}{mathx}{m}{n}{
  <5> <6> <7> <8> <9> <10>
  <10.95> <12> <14.4> <17.28> <20.74> <24.88>
  mathx10
}{}
\DeclareSymbolFont{mathx}{U}{mathx}{m}{n}
\DeclareMathSymbol{\bigtimes}{1}{mathx}{"91}

\usepackage{setspace}
\allowdisplaybreaks

\usepackage{graphicx}
\graphicspath{{images/}}
\usepackage{tabularx} 

\usepackage{tikz}
\usepackage{tikz-cd}
\usetikzlibrary{knots}

\usepackage[all]{xy}
\entrymodifiers={!!<0pt,0.7ex>+}

\usepackage{xcolor}
\usepackage{enumitem}
\usepackage{booktabs}
\usepackage{afterpage}
\usepackage{shuffle}

\newtheorem{theorem}{Theorem}[section]

\newtheorem{corollary}[theorem]{Corollary}
\newtheorem{definition}[theorem]{Definition}
\newtheorem{axiom}{Axiom}[section]
\newtheorem{example}[theorem]{Example}
\newtheorem{remark}[theorem]{Remark}

\newcommand{\Hilbert}{\mathcal{H}}
\newcommand{\Proj}{\mathcal{P}}

\newcommand{\eq}[1][r]{%
  \ar@<-3pt>@{-}[#1]
  \ar@<-1pt>@{}[#1]|<{}="gauche"
  \ar@<+0pt>@{}[#1]|-{}="milieu"
  \ar@<+1pt>@{}[#1]|>{}="droite"
  \ar@/^2pt/@{-}"gauche";"milieu"
  \ar@/_2pt/@{-}"milieu";"droite"
}

\newcommand{\bigon}[4][r]{%
  \ar@/^1pc/[#1]^{#2}_*=<0.3pt>{}="HAUT"
  \ar@/_1pc/[#1]_{#3}^*=<0.3pt>{}="BAS"
  \ar@{=>} "HAUT";"BAS" ^{#4}
}

\newcommand{\bigons}[6][r]{%
  \ar@/^2pc/[#1]^{#2}_*=<0.3pt>{}="HAUT"
  \ar@{}[#1]^*=<0.3pt>{}="MIDH"_*=<0.3pt>{}="MIDB"
  \ar[#1]_(0.3){#3}
  \ar@/_2pc/[#1]_{#4}^*=<0.3pt>{}="BAS"
  \ar@{=>} "HAUT";"MIDH" ^{#5}
  \ar@{=>} "MIDB";"BAS" ^{#6}
}
\usepackage{natbib}
\usepackage{url}
\usepackage{hyperref}

\begin{document}
\title{Axiomatic Foundation of Quantum-Inspired Distance Metrics}
\author{Maryam Bagherian}
\email{maryambagherian@isu.edu} 
\affiliation{Department of Mathematics and Statistics, \\Idaho State University}

\begin{abstract}
We develop a comprehensive axiomatic framework for quantum-inspired distance metrics on projective Hilbert spaces, providing a unified foundation that organizes and generalizes existing measures in quantum information theory. Starting from five fundamental axioms—projective invariance, unitary covariance, superposition sensitivity, entanglement awareness, and measurement contextuality—we show that any admissible distance depends solely on state overlap and establish the uniqueness of the Fubini–Study metric as the canonical geodesic distance. Our framework further yields a hierarchy of comparison results relating the Fubini–Study metric, Bures distance, Euclidean distance, measurement-based pseudometrics, and entanglement-sensitive distances. Key contributions include an entanglement-geometry complementarity principle, high-dimensional concentration bounds, and operational interpretations connecting distances to state discrimination and quantum metrology. This work places the geometry of quantum state spaces on a rigorous axiomatic footing, bridging abstract metric theory, information geometry, and operational quantum principles.
\end{abstract}
\maketitle

\section{Introduction}

The mathematical foundation of quantum mechanics is built upon the postulate that the state space of any isolated physical system is a Hilbert space $\Hilbert$. A system at a given instant is fully described by a {state vector} $|\psi\rangle \in \Hilbert$, which is a unit vector defined up to a global phase $e^{i\theta}$, for some $\theta \in \mathbb{R}$ \cite{nielsen2010quantum}.

For an $n$-partite system, the Hilbert space is a tensor product $\Hilbert_1 \otimes \cdots \otimes \Hilbert_n$, and any pure state can be expanded in an orthonormal basis $\{|i\rangle\}$ as
\begin{equation}\notag
|\psi\rangle = \sum_{i=1}^{d} \alpha_i |i\rangle, \quad \sum_i |\alpha_i|^2 = 1,
\end{equation}
with complex coefficients $\alpha_i \in \mathbb{C}$ \cite{nielsen2010quantum, hayashi2017quantum}. From a data-science perspective, this is equivalent to embedding classical data points into a high-dimensional complex vector space, where quantum amplitudes encode richer structural information than classical feature vectors \cite{schuld2018supervised, wittek2014quantum}.

A central feature of multipartite systems is the distinction between separable and entangled states:

\begin{definition}[Separable and Entangled States]
A state $|\psi\rangle \in \Hilbert_1 \otimes \cdots \otimes \Hilbert_n$ is \emph{separable} if it can be written as $|\psi\rangle = |v_1\rangle \otimes \cdots \otimes |v_n\rangle$ for some $|v_i\rangle \in \Hilbert_i$. Otherwise, it is \emph{entangled}.
\end{definition}

Understanding distances between quantum states is crucial for tasks ranging from state discrimination and quantum metrology to quantum machine learning. Existing metrics, such as trace distance, fidelity, Bures distance, and the Fubini–Study metric, capture aspects of distinguishability and geometry \cite{wilde2013, bengtsson2017geometry}, but they have largely been studied in isolation, without a unifying mathematical framework.

In this work, we introduce an axiomatic framework for quantum-inspired distance metric functions. These axioms are grounded in fundamental quantum principles:
\begin{enumerate}[label=(\roman*)]
    \item {Superposition:} Distances must account for both amplitudes and relative phases in coherent superpositions.
    \item {Entanglement:} Measures should reflect non-classical correlations in composite systems.
    \item {Measurement Contextuality:} Distances may depend on the basis or measurement context, consistent with probabilistic collapse.
    \item {Phase Sensitivity:} While global phases are unobservable, relative phases are physically significant and must influence the distance.
\end{enumerate}
By formalizing these principles, our framework allows us to (i) characterize existing quantum distances within a unified axiomatic system, (ii) prove the uniqueness of the Fubini–Study metric under natural geometric constraints, and (iii) establish new comparison inequalities, entanglement-sensitive bounds, and operational interpretations relevant to quantum information tasks.

This approach provides both conceptual clarity and practical guidance for constructing distance measures that respect the intrinsic geometry and physics of quantum state spaces, with applications spanning quantum information theory, quantum metrology, and quantum machine learning. 

The rest of this manuscript is organized as the following: Section~\ref{sec:found} established the necessary mathematical framework for quantum state spaces and Section~\ref{sec:axiom} outlines the proposed axiomatic framework for characterization of quantum-inspired distances, followed by Section~\ref{sec:thm} where we present the central mathematical results of this manuscript. In Section~\ref{sec:rel}, we highlight the connections of the proposed framework to classical axiomatic distance theory, and we conclude in Section~\ref{sec:open} to outline the limitations of the proposed framework and pose a few open problems for future directions.

\section{Foundations of Quantum-Inspired Distance Metric Functions}
\label{sec:found}

This section provides the mathematical framework underlying distance
functions on quantum state spaces. We first recall the structure of
pure-state quantum mechanics in finite-dimensional Hilbert spaces.
We then outline structural constraints that any quantum-inspired
distance function must satisfy, and finally organize the principal
families of quantum distances into a coherent hierarchy.

\subsection{The Geometry of Pure Quantum States}

Throughout this section, let $\mathcal{H}$ be a finite-dimensional
complex Hilbert space of dimension $d$.
The mathematical formulation of quantum mechanics is based on
Hilbert spaces equipped with an inner product
\cite{vonNeumann1932,nielsen2010quantum,bengtsson2017geometry}.

\begin{definition}[Hilbert Space]
A \emph{Hilbert space} $\mathcal{H}$ is a complex vector space endowed
with an inner product
$\langle \cdot,\cdot \rangle : \mathcal{H}\times\mathcal{H}\to\mathbb{C}$
that is linear in its second argument, conjugate symmetric,
and positive definite. The induced norm is
\[
\|\psi\| = \sqrt{\langle\psi,\psi\rangle}.
\]
\end{definition}

Physical pure states are not vectors themselves but rays in
$\mathcal{H}$, reflecting global phase invariance
\cite{vonNeumann1932,peres1997quantum}.

\begin{definition}[Projective Hilbert Space]
Two nonzero vectors $\ket{\psi},\ket{\varphi}\in\mathcal{H}$
are equivalent if
\(
\ket{\varphi}=e^{i\theta}\ket{\psi},\) for some \(\theta\in\mathbb{R}.\)
The space of equivalence classes is the
\emph{projective Hilbert space}
\begin{equation}\notag
\Proj(\mathcal{H})
=
(\mathcal{H}\setminus\{0\})/\sim .
\end{equation}
Elements $[\psi]\in\Proj(\mathcal{H})$ represent pure quantum states.
\end{definition}

In finite dimension,
$\Proj(\mathcal{H})\cong\mathbb{C}P^{d-1}$,
a compact Kähler manifold
\cite{bengtsson2017geometry}. For composite systems,
$\mathcal{H}=\mathcal{H}_A\otimes\mathcal{H}_B$,
and pure states may be separable or entangled
\cite{nielsen2010quantum}.

\begin{definition}[Separable and Entangled States]
A pure state $\ket{\psi}\in\mathcal{H}_A\otimes\mathcal{H}_B$
is \emph{separable} if
\(\ket{\psi}=\ket{\psi_A}\otimes\ket{\psi_B}.
\)
Otherwise, it is \emph{entangled}.
\end{definition}

\subsection{Structural Constraints on Quantum Distance Functions}

Let
\(
d:\Proj(\mathcal{H})\times\Proj(\mathcal{H})
\to\mathbb{R}_{\ge 0}
\)
be a candidate quantum distance function.
Fundamental principles of quantum mechanics impose structural
constraints on admissible distances.

\paragraph{(i) Projective and Unitary Invariance.} Physical predictions are invariant under global phase and
unitary evolution \cite{peres1997quantum,nielsen2010quantum}.
Thus any admissible distance must satisfy
\begin{align}\notag
d([e^{i\theta}\psi],[\varphi])
&=
d([\psi],[\varphi]),\\\notag
d([U\psi],[U\varphi])
&=
d([\psi],[\varphi]),
\qquad
\forall U\in\mathcal{U}(\mathcal{H}).
\end{align}

\paragraph{(ii) Operational Consistency.}

Quantum measurements are described by positive operator-valued measures (POVMs)
$\{M_m\}$ satisfying
\(\sum_m M_m^\dagger M_m = \mathbb{I}\) \cite{Holevo2011,nielsen2010quantum}.

The \emph{Born rule} prescribes the probability of outcome \(m\) when measuring
state $\rho_\psi$
\[
p_\psi(m) = \operatorname{Tr}(\rho_\psi M_m^\dagger M_m).
\]
Two states are said to be \emph{operationally indistinguishable} if they
induce identical probability distributions for all POVMs, i.e., \(p_\psi(m) = p_\varphi(m) \quad \forall m, \forall \text{POVMs}.
\)
Accordingly, a quantum distance function must satisfy
\(
d([\psi],[\varphi]) = 0\) if and only if the states are operationally indistinguishable.

\paragraph{(iii) Contractivity Under Quantum Channels.}

Let $d$ be a distance function on the set of quantum states
$\mathcal{D}(\mathcal{H})$ (density matrices). Physical evolution
of mixed states is described by completely positive trace-preserving
(CPTP) maps $\Phi$. The principle that information-processing cannot increase
distinguishability requires that any admissible distance satisfy
\cite{petz1996monotone}
\[
d(\Phi(\rho),\Phi(\sigma)) \le d(\rho,\sigma), \quad
\forall \rho,\sigma \in \mathcal{D}(\mathcal{H}).
\]
Distances satisfying this property are called \emph{monotone quantum metrics}.

\subsection{A Hierarchy of Quantum Distance Functions}

Distance measures on $\Proj(\mathcal{H})$ can be organized
according to the structural principles they encode
\cite{bengtsson2017geometry,nielsen2010quantum}.

\paragraph{(i) Geometric Distances}
These distances arise from the differential geometry of projective space.

\begin{definition}[Hilbert-Space Distance]
For normalized vectors,
\[
d_{\mathrm{HS}}(\psi,\varphi)
=
\|\ket{\psi}-\ket{\varphi}\|.
\]
Although this is a metric on the unit sphere, it is not phase-invariant
and therefore does not descend to $\Proj(\mathcal{H})$.
\end{definition}

\begin{definition}[Transition Probability]
\[
P(\psi,\varphi)=|\langle\psi|\varphi\rangle|^2
\]
is phase-invariant and equals the Born probability of obtaining $\ket{\varphi}$ when measuring $\ket{\psi}$
\cite{peres1997quantum}.
\end{definition}

\begin{definition}[Fubini--Study metric]\label{def:fubini_study}
The \emph{Fubini--Study metric} is the unique (up to an overall scale)
unitarily invariant K\"ahler Riemannian metric on $\mathcal P(\mathcal H)$. For a normalized representative $\psi \in \mathcal H$ 
($\langle \psi|\psi\rangle = 1$), the metric tensor is given by
\[
g_{\mathrm{FS},\psi}(\delta\psi,\delta\psi)
=
\langle \delta\psi | \delta\psi \rangle
-
|\langle \psi | \delta\psi \rangle|^2,
\]
where $\delta\psi$ is taken orthogonal to $\psi$
(i.e.\ $\langle \psi|\delta\psi\rangle = 0$). The induced geodesic distance between rays
$[\psi],[\varphi] \in \mathcal P(\mathcal H)$ is
\[
d_{\mathrm{FS}}([\psi],[\varphi])
=
\arccos\!\left(
\frac{|\langle\psi|\varphi\rangle|}
{\|\psi\|\,\|\varphi\|}
\right),
\qquad
0 \le d_{\mathrm{FS}} \le \frac{\pi}{2}.
\]
In particular, for normalized vectors,
\(
d_{\mathrm{FS}}([\psi],[\varphi])
=
\arccos\big(|\langle\psi|\varphi\rangle|\big)
\) \cite{bengtsson2017geometry}.
\end{definition}

\paragraph{(ii) Operational Distances}

\begin{definition}[Trace Distance]
For quantum states $\rho,\sigma \in \mathcal{D}(\mathcal{H})$, the \emph{trace distance} is
\cite{nielsen2010quantum,Holevo2011}:
\[
D_{\mathrm{tr}}(\rho,\sigma) = \frac12 \|\rho-\sigma\|_1.
\]
It quantifies the operational distinguishability of two states via measurement outcomes.
\end{definition}

\begin{remark}[Trace Distance for Pure States and Helstrom Bound]
For pure states $\rho_\psi=\ket{\psi}\!\bra{\psi}$ and $\rho_\varphi=\ket{\varphi}\!\bra{\varphi}$:
\[
D_{\mathrm{tr}}(\rho_\psi,\rho_\varphi) = \sqrt{1 - |\langle\psi|\varphi\rangle|^2}.
\]
Helstrom’s theorem \cite{helstrom1969quantum} states that, for equal prior probabilities, the optimal success probability of distinguishing these states is
\[
P_{\mathrm{succ}} = \frac12 \left(1 + D_{\mathrm{tr}}(\rho_\psi,\rho_\varphi)\right)
= \frac12 \left(1 + \sin d_{\mathrm{FS}}([\psi],[\varphi])\right),
\]
showing that operational distinguishability is a monotone function of the Fubini–Study distance.
\end{remark}

\paragraph{(iii) Fidelity and Bures Distance}

\begin{definition}[Fidelity]
For pure states, the \emph{fidelity} is defined as \cite{Uhlmann1976}:
\[
F(\psi,\varphi) = |\langle\psi|\varphi\rangle|^2.
\]
Fidelity measures the “closeness” of two quantum states, taking values in $[0,1]$.
\end{definition}

\begin{definition}[Bures Distance]\label{def:bures}
The \emph{Bures distance} between pure states is defined as \cite{Bures1969}:
\[
d_{\mathrm{B}}(\psi,\varphi) = \sqrt{2\left(1-|\langle\psi|\varphi\rangle|\right)}
= 2\sin\frac{d_{\mathrm{FS}}([\psi],[\varphi])}{2}.
\]
It is a proper metric on the projective Hilbert space and is unitarily invariant.
\end{definition}

\begin{remark}[Relation to Trace Distance]
Fidelity and trace distance satisfy the Fuchs--van de Graaf inequalities \cite{fuchs2002cryptographic}:
\[
1-\sqrt{F(\rho,\sigma)} \le D_{\mathrm{tr}}(\rho,\sigma) \le \sqrt{1-F(\rho,\sigma)}.
\]
\end{remark}

\paragraph{(iv) Entanglement-Sensitive Distances}

\begin{definition}[Entanglement of Bipartite Pure States]
For a bipartite pure state $\ket{\psi} \in \mathcal{H}_A \otimes \mathcal{H}_B$,
the \emph{entanglement} is quantified by the von Neumann entropy of the reduced state \cite{nielsen2010quantum}:
\begin{align}
E(\psi) &= S(\rho_A) = -\operatorname{Tr}(\rho_A \log \rho_A),
\nonumber \\
\rho_A &= \operatorname{Tr}_B \ket{\psi}\!\bra{\psi}.\nonumber 
\end{align}
\end{definition}

\begin{definition}[Entanglement-Sensitive Distance]
Given two bipartite pure states $[\psi],[\varphi] \in \Proj(\mathcal{H}_A \otimes \mathcal{H}_B)$,
the \emph{entanglement-sensitive distance} is
\[
d_E([\psi],[\varphi]) =
\sqrt{ d_{\mathrm{FS}}([\psi],[\varphi])^2 + |E(\psi)-E(\varphi)|^2 }.
\]
\end{definition}

\begin{remark}[Metric Properties]
Since the von Neumann entropy is Lipschitz-continuous with respect to the trace distance
(Fannes--Audenaert inequality \cite{fannes1973continuity,audenaert2007sharp}) and
trace distance is equivalent to Bures distance in finite dimensions,
$d_E$ defines a proper metric on the space of finite-dimensional bipartite pure states.
\end{remark}


\paragraph{(v) Computational Considerations}

For general $n$-qubit states,
inner-product computation scales as $O(2^n)$.
For matrix product states with bond dimension $\chi$,
tensor contraction scales as $O(n\chi^3)$
\cite{Orus2014}.
Estimating fidelity to precision $\epsilon$
requires $O(1/\epsilon^2)$ samples by standard
concentration bounds.


This hierarchy reflects progressively stronger incorporation of
quantum-mechanical structure: from projective symmetry,
to operational distinguishability, to full
information-theoretic consistency.

\subsection{Illustrative Examples}
\begin{example}[Quantum State Discrimination]

The operational significance of distances is illustrated by
state discrimination. For two pure states $\rho_\psi$ and $\rho_\phi$
prepared with equal prior probabilities, Helstrom's theorem
\cite{helstrom1969quantum} implies that the optimal success probability is
\[
P_{\mathrm{succ}}
=
\frac12 \bigl(1 + D_{\mathrm{tr}}(\rho_\psi, \rho_\varphi)\bigr).
\]
For pure states,
\[
D_{\mathrm{tr}}(\rho_\psi, \rho_\varphi)
=
\sqrt{1-F(\psi,\varphi)}
=
\sin d_{\mathrm{FS}}([\psi],[\varphi]),
\]
showing that distinguishability is a monotone function of
the Fubini--Study distance.
\end{example}
\begin{example}[Quantum Machine Learning]
In quantum kernel methods, a classical input $x \in \mathcal{X}$ is
mapped into a quantum state $\ket{\phi(x)}\in\mathcal{H}$, inducing
the kernel
\[
K(x,y)=|\langle\varphi(x)|\varphi(y)\rangle|^2,
\]
which coincides with quantum fidelity and provides a geometric measure
of similarity \cite{Havlicek2019,Schuld2019}. Distance measures of this
type also play a role in quantum‐inspired compression and
dimensionality‐reduction schemes for hybrid classical–quantum data
processing in edge computing settings \cite{bagherian2023classical}.
\end{example}

\section{Axiomatic Framework for Quantum-Inspired Distance Metric  Functions}\label{sec:axiom}

Having surveyed the mathematical structures of quantum state spaces and existing quantum-inspired distance measures, we now present our main contribution: an axiomatic framework that characterizes genuinely quantum-inspired distances. This framework has three objectives: (i) provide necessary and sufficient conditions for a distance to be quantum-inspired, (ii) reveal the consequences of each quantum principle, and (iii) classify distances according to the quantum phenomena they capture.

\subsection{Preliminary Definitions and Notation}

Let $\Hilbert$ be a complex Hilbert space with $\dim(\Hilbert) = d < \infty$, and let $\Proj(\Hilbert)$ denote the projective Hilbert space of pure states, i.e., equivalence classes of unit vectors modulo global phase. For $\ket{\psi} \in \Hilbert$ with $\norm{\psi} = 1$, its equivalence class is denoted $[\psi] \in \Proj(\Hilbert)$. If $\ket{\psi} \sim e^{i\theta}\ket{\psi}$ and 
$\ket{\varphi} \sim e^{i\phi}\ket{\varphi}$, then
\(
\braket{\psi}{\varphi}
\;\longmapsto\;
e^{-i\theta} e^{i\phi}\,\braket{\psi}{\varphi},
\)
so the inner product transforms by a phase factor. However, its modulus
\(
\left|\braket{\psi}{\varphi}\right|
\)
remains unchanged. Hence $\left|\braket{\psi}{\varphi}\right|$ is well-defined on $\Proj(\Hilbert)$.

\begin{definition}[Quantum-Inspired Distance Function]
A \emph{quantum-inspired distance function} is a map
\[
d: \Proj(\Hilbert) \times \Proj(\Hilbert) \to \mathbb{R}_{\geq 0},
\]
quantifying the dissimilarity between two states. Metric properties are not assumed a priori but are derived or imposed as appropriate. For composite systems, we consider $\Hilbert = \bigotimes_{k=1}^n \Hilbert_k$, $\dim(\Hilbert_k) = d_k$.
\end{definition}

\begin{definition}[Basic Metric Properties]
A distance $d$ on $\Proj(\Hilbert)$ is:
\begin{enumerate}[label=(\roman*)]
    \item \emph{Non-negative:} $d([\psi],[\varphi]) \ge 0$.
    \item \emph{Symmetric:} $d([\psi],[\varphi]) = d([\varphi],[\psi])$.
    \item \emph{Positive-definite:} $d([\psi],[\varphi]) = 0 \implies [\psi] = [\varphi]$.
    \item \emph{Metric:} satisfies (i)--(iii) and the triangle inequality.
\end{enumerate}
\end{definition}

\subsection{Fundamental Axioms}\label{subsec:ax}

We present a streamlined set of axioms for quantum-inspired distances. 
They are divided into three layers: (i) intrinsic geometric properties, 
(ii) composite-system refinements, and (iii) operational (measurement-based) distances.

\subsubsection{Intrinsic Geometric Axioms}

\begin{axiom}[Ray Well-Definedness]
\label{ax:ray}
The distance is well-defined on rays: for any $\ket{\psi} \sim e^{i\theta}\ket{\psi}$ and $\ket{\varphi} \sim e^{i\phi}\ket{\varphi}$,
\(
d([\psi],[\varphi])\) is independent of the choice of representatives.
\end{axiom}

\begin{axiom}[Unitary Invariance]
\label{ax:unitary_invariance}
For all rays $[\psi],[\varphi] \in \Proj(\Hilbert)$ and any unitary $U$ on $\Hilbert$,
\[
d(U[\psi],U[\varphi]) = d([\psi],[\varphi]).
\]
\end{axiom}

\begin{remark}[Unitary Invariance vs. Covariance]
Axiom~\ref{ax:unitary_invariance} asserts \emph{unitary invariance}: the distance between two rays is unchanged under the same unitary transformation. This differs from \emph{unitary covariance}, where the distance transforms consistently under different unitaries applied to each argument. Invariance is natural because quantum-inspired distances should respect global basis changes.
\end{remark}

\begin{axiom}[Superposition Sensitivity]
\label{ax:superposition}
Let $[\psi_1]$ and $[\psi_2]$ be orthogonal rays, and let $\alpha, \beta \in \mathbb{C}$ with $|\alpha|^2 + |\beta|^2 = 1$. Then coherent superpositions are distinguishable:
\[
d([\alpha \psi_1 + \beta \psi_2], [\alpha' \psi_1 + \beta' \psi_2]) > 0
\]
for $\alpha \neq \alpha'$ or $\beta \neq \beta'$, up to global phase, even if the classical probabilities $|\alpha|^2$ and $|\beta|^2$ coincide.
\end{axiom}

\begin{axiom}[Non-Degeneracy]
\label{ax:nondegeneracy}
\[
d([\psi],[\varphi]) = 0 \iff [\psi] = [\varphi].
\]
\end{axiom}

\begin{axiom}[Metric Compatibility]
\label{ax:triangle}
If $d$ is a metric, it satisfies the triangle inequality:
\[
d([\psi],[\varphi]) \le d([\psi],[\chi]) + d([\chi],[\varphi]), \quad \forall [\psi],[\varphi],[\chi] \in \Proj(\Hilbert).
\]
\end{axiom}

\begin{axiom}[Two-Level Geodesic Additivity]
\label{ax:geodesic}
For any two-dimensional subspace of $\Hilbert$, the distance depends only on the transition probability $|\langle \psi|\varphi\rangle|$.  
Moreover, if three rays $[\psi_1],[\psi_2],[\psi_3]$ lie along a geodesic in this subspace with
\[
|\langle \psi_1|\psi_3\rangle| = |\langle \psi_1|\psi_2\rangle| \, |\langle \psi_2|\psi_3\rangle|,
\]
then
\[
d([\psi_1],[\psi_3]) = d([\psi_1],[\psi_2]) + d([\psi_2],[\psi_3]).
\]
\end{axiom}

\begin{remark}[Uniqueness of the Fubini--Study Metric]
\label{rem:fs_unique}
Axiom~\ref{ax:geodesic} is optional but has a remarkable consequence: when combined with 
Ray Well-Definedness (Axiom~\ref{ax:ray}), Unitary Invariance (Axiom~\ref{ax:unitary_invariance}), 
Non-Degeneracy (Axiom~\ref{ax:nondegeneracy}), and Metric Compatibility (Axiom~\ref{ax:triangle}), 
it uniquely fixes the distance, up to a positive multiplicative constant, to be the \emph{Fubini--Study metric}:
\[
d_{\mathrm{FS}}([\psi],[\varphi]) = \arccos|\langle \psi|\varphi\rangle|, \quad 0 \le d_{\mathrm{FS}} \le \frac{\pi}{2}.
\]
Without Axiom~\ref{ax:geodesic}, other unitarily invariant distances on projective Hilbert space remain admissible.
\end{remark}

\subsubsection{Composite-System Refinement}

\begin{axiom}[Entanglement Awareness]
\label{ax:entanglement_awareness}
Let $\Hilbert = \Hilbert_A \otimes \Hilbert_B$ with $\dim(\Hilbert_A), \dim(\Hilbert_B) \ge 2$. 
There exist states $[\psi],[\varphi]$ such that
\begin{eqnarray}\notag
  \tr_B(\ketbra{\psi}{\psi}) &=& \tr_B(\ketbra{\varphi}{\varphi}), \\ \notag
  \tr_A(\ketbra{\psi}{\psi}) &=& \tr_A(\ketbra{\varphi}{\varphi}),
\end{eqnarray}
yet 
\[
d([\psi],[\varphi]) > 0.
\]
\end{axiom}
\begin{remark}
Axiom~\ref{ax:entanglement_awareness} asserts that the global pure state of a composite system is not determined by its local reduced states. 
Even if two states $[\psi]$ and $[\varphi]$ have identical marginals on both subsystems,
\(
\tr_B(\ketbra{\psi}{\psi}) = \tr_B(\ketbra{\varphi}{\varphi}),\) and \(\tr_A(\ketbra{\psi}{\psi}) = \tr_A(\ketbra{\varphi}{\varphi}),
\)
they may nevertheless represent distinct global states. 
Thus the map
\[
[\psi] \longmapsto 
\left(
\tr_B(\ketbra{\psi}{\psi}),
\tr_A(\ketbra{\psi}{\psi})
\right),
\]
is not injective. Conceptually, this expresses that entanglement contains relational information that is invisible to the individual subsystems. 
All local measurement statistics may coincide, while global correlations still differ. 
In particular, relative phases between Schmidt components cannot, in general, be reconstructed from the reduced states alone \cite{Holevo2011}.
\end{remark}

\subsubsection{Operational / Measurement Axioms}

 \begin{axiom}[Measurement Contextuality]
\label{ax:measurement_contextuality}
Let $\mathcal{M} = \{M_m\}$ be a POVM on $\Hilbert$. 
The measurement-induced distance $d_\mathcal{M}$ depends only on the induced probability distributions
\[
p_\mathcal{M}(\psi) = \{\bra{\psi} M_m^\dagger M_m \ket{\psi}\}_m.
\]
In particular,
\[
d_\mathcal{M}([\psi],[\varphi]) = 0 
\quad \Longleftrightarrow \quad 
p_\mathcal{M}(\psi) = p_\mathcal{M}(\varphi).
\]
Distinct POVMs may induce different distances; in particular, if $\mathcal{M}$ is not informationally complete, there exist $[\psi] \neq [\varphi]$ with $d_\mathcal{M}([\psi],[\varphi]) = 0$.
\end{axiom}

\begin{remark}[Informationally Complete POVMs]
A POVM $\mathcal{M} = \{M_m\}$ on a finite-dimensional Hilbert space $\Hilbert$ is called \emph{informationally complete} if the measurement statistics it produces uniquely determine any quantum state. That is, the map
\(
\rho \;\mapsto\; \{ \mathrm{Tr}(M_m \rho) \}_m
\)
is injective on the set of density operators. Equivalently, no two distinct quantum states yield the same probability distribution under $\mathcal{M}$. This property implies that, given sufficiently many measurement outcomes, one can \emph{reconstruct} the unknown state $\rho$ from the observed frequencies \cite{nielsen2010quantum}.
\end{remark}


\noindent In short, Axioms \ref{ax:ray} and \ref{ax:unitary_invariance} encode fundamental quantum structure, i.e., phase irrelevance and unitary invariance. 
Axiom \ref{ax:superposition} distinguishes quantum from classical probability by ensuring distances recognize coherent superpositions. 
Axiom \ref{ax:entanglement_awareness} applies to multipartite systems and captures global correlations that are invisible to local reductions, thereby incorporating entanglement as a genuinely nonclassical structural feature. 
Axiom \ref{ax:measurement_contextuality} addresses operational, measurement-defined distances, typically yielding pseudometrics when the measurement is not informationally complete. 

Together with non-degeneracy and metric compatibility (Axioms \ref{ax:nondegeneracy} and \ref{ax:triangle}, respectively), these conditions delineate a hierarchy of quantum-inspired distance functions. 
Intrinsic geometric axioms characterize distances on projective Hilbert space, composite-system refinements incorporate entanglement structure, and operational axioms relate distances to experimentally accessible statistics. 
Different quantum distances satisfy different subsets of these axioms, thereby situating them within this structural hierarchy.
\begin{definition}[Quantum-Inspired Distance Function]\label{def:qidf}
A function
\[
d : \Proj(\Hilbert) \times \Proj(\Hilbert)
\to \mathbb{R}_{\ge 0},
\]
is called a \emph{quantum-inspired distance function}
if it satisfies Axioms \ref{ax:ray}, 
\ref{ax:unitary_invariance}, and \ref{ax:superposition}.
\end{definition}

\begin{definition}
A quantum-inspired distance function given in Definition~\ref{def:qidf} satisfying additionally:
\begin{enumerate}[label=(\roman*)]
    \item Axioms \ref{ax:nondegeneracy} and \ref{ax:triangle}
    is called a \emph{quantum-inspired metric};
    
    \item Axiom \ref{ax:entanglement_awareness}, in the case
    $\Hilbert = \Hilbert_A \otimes \Hilbert_B$,
    is called \emph{entanglement-aware};
    
    \item Axiom \ref{ax:measurement_contextuality}, when defined relative to a measurement context,
    is called \emph{measurement-contextual}.
\end{enumerate}
\end{definition}

\subsection{Immediate Consequences}
The following results are immediate consequences of Axioms~\ref{ax:ray}--\ref{ax:measurement_contextuality}.
\begin{theorem}[Overlap Dependence]
\label{thm:overlap}
Let $d$ be a mapping on $\Proj(\Hilbert)$ satisfying
Axiom~\ref{ax:ray}
and Axiom~\ref{ax:unitary_invariance}.
Then there exists a function
\[
f : [0,1] \to \mathbb{R}_{\ge 0},
\]
such that for all nonzero $\psi,\varphi \in \Hilbert$,
\[
d([\psi],[\varphi])
=
f\!\left(
\frac{|\langle \psi|\varphi\rangle|}
{\|\psi\|\,\|\varphi\|}
\right).
\]
\end{theorem}
\begin{proof}
Let $\psi,\varphi$ be normalized vectors. 
By unitary invariance (Axiom~\ref{ax:unitary_invariance}), the value of $d([\psi],[\varphi])$
is unchanged under any unitary transformation applied to both states.
Given two pairs of normalized vectors
$(\psi,\varphi)$ and $(\psi',\varphi')$
satisfying
\[
|\langle\psi|\varphi\rangle|
=
|\langle\psi'|\varphi'\rangle|,
\]
there exists a unitary $U$ such that
\[
U\psi=\psi',
\qquad
U\varphi=\varphi'.
\]
Indeed, by choosing orthonormal bases extending each pair,
one may construct $U$ mapping one basis to the other.
Hence,
\[
d([\psi],[\varphi])
=
d(U[\psi],U[\varphi])
=
d([\psi'],[\varphi']).
\]
Therefore, the distance depends only on the quantity
$|\langle\psi|\varphi\rangle|$.
Ray well-definedness (Axiom~\ref{ax:ray}) allows removal of normalization,
yielding the stated form.
\end{proof}

\begin{theorem}[Monotonicity]
\label{thm:monotonicity}
Let $d$ be a quantum-inspired distance function satisfying two-level geodesic additivity (Axiom~\ref{ax:geodesic}), and let 
\(
d([\psi],[\varphi]) = f(|\braket{\psi}{\varphi}|).
\)
Then the function $f:[0,1]\to\mathbb{R}_{\ge 0}$ is monotonically decreasing.
\end{theorem}

\begin{proof}
Restrict to a two-dimensional subspace spanned by an orthonormal pair $\{\ket{0},\ket{1}\}$, 
and parametrize states along a geodesic by
\[
\ket{\psi(\theta)} = \cos\theta\,\ket{0} + \sin\theta\,\ket{1}, 
\quad \theta \in [0,\pi/2].
\]
Then 
\(
|\braket{\psi(0)}{\psi(\theta)}| = \cos\theta,
\)
which decreases as \(\theta\) increases. Geodesic additivity implies that $d([\psi(0)],[\psi(\theta)])$ increases with $\theta$.  
Hence $f(r)$ must decrease as $r = |\braket{\psi(0)}{\psi(\theta)}|$ increases.  
Therefore, $f$ is monotonically decreasing on $[0,1]$.
\end{proof}

\begin{theorem}[Classification of Quantum-Inspired Metrics]  
\label{thm:classification}  
Let $d$ be a distance function on the projective Hilbert space $\Proj(\Hilbert)$ satisfying the following axioms and property:  
\begin{enumerate}[label=(\roman*)]  
    \item Axiom~\ref{ax:ray}: Ray well-definedness,  
    \item Axiom~\ref{ax:unitary_invariance}: Unitary invariance,  
    \item Axiom~\ref{ax:geodesic}: Two-level geodesic additivity,  
    \item The triangle inequality.  
\end{enumerate}  
Then there exists a function $f:[0,1] \to \mathbb{R}_{\ge 0}$ such that  
\[
d([\psi],[\varphi]) = f\big(|\braket{\psi}{\varphi}|\big), \quad \forall [\psi],[\varphi] \in \Proj(\Hilbert),
\]  
with the following properties:  
\begin{enumerate}[label=(\roman*)]  
    \item $f(1) = 0$ and $f(x) > 0$ for all $x < 1$,  
    \item $f$ is strictly decreasing on $[0,1]$,  
    \item Defining $g(\theta) := f(\cos\theta)$ for $\theta \in [0,\pi/2]$, the function $g$ is subadditive, i.e.,  
    \[
    g(\theta_1 + \theta_2) \le g(\theta_1) + g(\theta_2), 
    \]  
\end{enumerate}  
$\forall \theta_1, \theta_2 \ge 0 \text{ with } \theta_1 + \theta_2 \le \pi/2.$ Conversely, any function $f$ satisfying these conditions defines a quantum-inspired metric via 
\[
d([\psi],[\varphi]) := f\big(|\braket{\psi}{\varphi}|\big),
\] 
that satisfies the above axioms.  
\end{theorem}

\begin{proof}
By Theorem~\ref{thm:overlap}, unitary invariance and ray well-definedness imply
\[
d([\psi],[\varphi]) = f(|\braket{\psi}{\varphi}|), \quad f:[0,1]\to\mathbb{R}_{\ge 0}.
\]
Positivity and normalization follow from the definition of a distance: $d([\psi],[\psi]) = 0$ gives $f(1) = 0$, and non-degeneracy implies $f(x) > 0$ for $x < 1$. \\
Monotonicity is a consequence of geodesic additivity (Theorem~\ref{thm:monotonicity}): restricting to a two-dimensional subspace spanned by $\ket{\psi}$ and $\ket{\varphi}$, the distance along a geodesic increases with the Fubini--Study angle $\theta$, while $|\braket{\psi(0)}{\psi(\theta)}|$ decreases with $\theta$. Hence $f$ is strictly decreasing in $|\braket{\psi}{\varphi}|$.\\
Subadditivity of $g(\theta) := f(\cos\theta)$ follows from the triangle inequality applied to three states $[\psi_0],[\psi_1],[\psi_2]$ lying along a geodesic with angles $\theta_1$ and $\theta_2$
\begin{multline}\nonumber 
    g(\theta_1+\theta_2) = d([\psi_0],[\psi_2]) \\
    \le d([\psi_0],[\psi_1]) + d([\psi_1],[\psi_2])
    = g(\theta_1) + g(\theta_2).
\end{multline}
Conversely, any function $f$ satisfying these properties defines a distance
\[
d([\psi],[\varphi]) = f(|\braket{\psi}{\varphi}|),
\]
which is ray well-defined, unitarily invariant, and satisfies both two-level geodesic additivity and the triangle inequality.
\end{proof}

\begin{corollary}[Canonical Examples]
The Fubini--Study and Bures metrics arise as specific choices of admissible functions $f$ in Theorem~\ref{thm:classification}, when restricted to pure states.
\end{corollary}

\begin{proof}
On pure states, the Fubini--Study metric is (see Definitions~\ref{def:fubini_study})  
\[
d_\mathrm{FS}([\psi],[\varphi]) = \arccos|\braket{\psi}{\varphi}|, \quad f_\mathrm{FS}(r) = \arccos r,
\]
and the Bures metric is (see Definitions~\ref{def:bures}) 
\[
d_\mathrm{Bures}([\psi],[\varphi]) = \sqrt{2 - 2|\braket{\psi}{\varphi}|}, \quad f_\mathrm{Bures}(r) = \sqrt{2-2r}.
\]
Both $f_\mathrm{FS}$ and $f_\mathrm{Bures}$ are functions of $|\braket{\psi}{\varphi}|$, satisfy ray well-definedness and unitary invariance, are nonnegative and strictly decreasing for $r \in [0,1]$, and their corresponding functions
\[
g_\mathrm{FS}(\theta) := f_\mathrm{FS}(\cos\theta), \quad g_\mathrm{Bures}(\theta) := f_\mathrm{Bures}(\cos\theta),
\]
are subadditive, ensuring the triangle inequality is satisfied. Hence, both are admissible functions $f$ in Theorem~\ref{thm:classification}.
\end{proof}

\begin{theorem}[Failure of Projective Invariance for Euclidean Distance]
The Euclidean distance on Hilbert space,
\[
d_E(\ket{\psi},\ket{\varphi}) = \|\ket{\psi} - \ket{\varphi}\|,
\]
is not invariant under global phase changes and therefore not quantum-inspired.
\end{theorem}

\begin{proof}
Let $\ket{\psi},\ket{\varphi} \in \Hilbert$ with $\ket{\psi} = \ket{\varphi}$. Then
\[
d_E(\ket{\psi},\ket{\varphi}) = \|\ket{\psi}-\ket{\varphi}\| = 0.
\]

Now consider a global phase change $\ket{\psi} \mapsto -\ket{\psi}$:
\[
d_E(-\ket{\psi},\ket{\varphi}) = \| -\ket{\psi}-\ket{\varphi}\| = \|-2\ket{\psi}\| = 2 \|\ket{\psi}\| \neq 0.
\]

Hence $d_E$ depends on the choice of representative of the ray $[\psi]$ and violates ray well-definedness (Axiom~\ref{ax:ray}). Therefore, it is not quantum-inspired.
\end{proof}

\begin{theorem}[Entanglement-Aware Distances]
Let $\Hilbert = \Hilbert_A \otimes \Hilbert_B$ with 
$\dim(\Hilbert_A),\dim(\Hilbert_B)\ge 2$. Define
\[
d_E([\psi],[\varphi]) 
= \sqrt{d_{\mathrm{FS}}([\psi],[\varphi])^2 
+ |E(\psi)-E(\varphi)|^2},
\]
where $d_{\mathrm{FS}}$ is the Fubini--Study distance and 
$E(\psi)$ is the von Neumann entropy of the reduced state 
$\tr_B(\ketbra{\psi}{\psi})$. 
Then $d_E$ satisfies Axiom~\ref{ax:entanglement_awareness}.
\end{theorem}

\begin{proof}
Let $[\psi],[\varphi]$ be bipartite pure states satisfying
\[
\tr_A(\ketbra{\psi}{\psi}) 
= 
\tr_A(\ketbra{\varphi}{\varphi}),
\qquad
\tr_B(\ketbra{\psi}{\psi}) 
= 
\tr_B(\ketbra{\varphi}{\varphi}),
\]
but assume $[\psi] \neq [\varphi]$. Such states exist; one may fix a Schmidt basis 
$\{\ket{i}_A \otimes \ket{i}_B\}$ and Schmidt coefficients 
$\{\lambda_i\}$, and consider
\[
\ket{\psi} = \sum_i \sqrt{\lambda_i}\,\ket{i}_A\ket{i}_B,
\qquad
\ket{\varphi} = \sum_i \sqrt{\lambda_i}\,e^{i\theta_i}
\ket{i}_A\ket{i}_B,
\]
where the phases are not all equal modulo a global phase.
These states have identical reduced density matrices but 
represent distinct rays whenever the phase vector 
$(\theta_i)$ is nontrivial.\\
Since $[\psi] \neq [\varphi]$, non-degeneracy of the 
Fubini--Study metric implies
\(
d_{\mathrm{FS}}([\psi],[\varphi]) > 0.
\)
Therefore
\[
d_E([\psi],[\varphi]) 
\ge d_{\mathrm{FS}}([\psi],[\varphi]) 
> 0.
\]
Thus there exist states with identical marginals but strictly positive distance, and Axiom~\ref{ax:entanglement_awareness} is satisfied.
\end{proof}

\begin{theorem}[Measurement Pseudometrics]
Let $\mathcal{M} = \{M_m\}$ be a POVM on $\Hilbert$. Define
\[
d_\mathcal{M}([\psi],[\varphi])
=
\sqrt{\sum_m 
\big(
\bra{\psi} M_m^\dagger M_m \ket{\psi}
-
\bra{\varphi} M_m^\dagger M_m \ket{\varphi}
\big)^2 }.
\]
Then $d_\mathcal{M}$ satisfies Axiom~\ref{ax:measurement_contextuality} and is a pseudometric.
\end{theorem}

\begin{proof}
Let
\[
p_\mathcal{M}(\psi) := \{\bra{\psi} M_m^\dagger M_m \ket{\psi}\}_m \in \mathbb{R}^n,
\]
be the measurement probability vector.  
By definition,
\(
d_\mathcal{M}([\psi],[\varphi]) = \| p_\mathcal{M}(\psi) - p_\mathcal{M}(\varphi) \|_2.
\)
Since the Euclidean norm is symmetric and satisfies the triangle inequality, $d_\mathcal{M}$ is symmetric and subadditive.\\
However, if the POVM $\mathcal{M}$ is not informationally complete, there exist distinct states $[\psi] \neq [\varphi]$ with identical measurement distributions $p_\mathcal{M}(\psi) = p_\mathcal{M}(\varphi)$. In such cases, $d_\mathcal{M}([\psi],[\varphi]) = 0$, showing that $d_\mathcal{M}$ is a \emph{pseudometric} rather than a true metric.  \\
By construction, $d_\mathcal{M}$ depends only on the measurement-induced probabilities, so it satisfies Axiom~\ref{ax:measurement_contextuality}.
\end{proof}


%
\section{Characterization Theorems and Relationships}\label{sec:thm}

Building upon the axiomatic framework introduced in Section 3, we now present the central mathematical results. We provide characterization theorems identifying which distances satisfy which axioms, uniqueness results for the Fubini–Study metric, rigorous comparison inequalities between different distance measures, and operational interpretations connecting abstract distances to concrete quantum information tasks. Collectively, these results demonstrate the coherence and power of the axiomatic approach.

\subsection{Characterization of Geometric Distances}\label{subsec:char}

The Fubini–Study distance occupies a canonical position among quantum-inspired distances: it is essentially the unique metric (up to monotonic transformation) satisfying the fundamental axioms and geodesically natural.
\begin{theorem}[Unitary Reduction]
\label{thm:unitary_reduction}
Let $d$ be a distance on $\Proj(\Hilbert)$ satisfying
Axiom~\ref{ax:ray} and Axiom~\ref{ax:unitary_invariance}.
Then there exists a function
\[
g : [0,\pi/2] \to \mathbb{R}_{\ge 0},
\]
such that for all rays $[\psi],[\varphi]$,
\[
d([\psi],[\varphi])
=
g\!\left(d_{\mathrm{FS}}([\psi],[\varphi])\right).
\]
If $d$ is a metric, then $g$ is strictly increasing and $g(0)=0$.
\end{theorem}
\begin{proof}
By Theorem~\ref{thm:overlap}, there exists
$f:[0,1]\to\mathbb{R}_{\ge 0}$ such that
\[
d([\psi],[\varphi])
=
f(|\braket{\psi}{\varphi}|).
\]
For normalized states,
\[
d_{\mathrm{FS}}([\psi],[\varphi])
=
\arccos |\braket{\psi}{\varphi}|.
\]
Define
\[
g(\theta) := f(\cos\theta),
\qquad \theta \in [0,\pi/2].
\]
Then
\[
d([\psi],[\varphi])
=
g(d_{\mathrm{FS}}([\psi],[\varphi])).
\]
If $d$ is a metric, then $d([\psi],[\psi])=0$ implies
$g(0)=0$, and non-degeneracy implies $g(\theta)>0$
for $\theta>0$. Monotonicity follows from
Theorem~\ref{thm:monotonicity}.
\end{proof}
\begin{theorem}[Additivity Rigidity]
\label{thm:additivity_rigidity}
Let $d$ be a metric on $\Proj(\Hilbert)$ satisfying
Axioms~\ref{ax:ray}, \ref{ax:unitary_invariance},
\ref{ax:triangle}, and \ref{ax:geodesic}.
Let $g$ be the function from
Theorem~\ref{thm:unitary_reduction}.
Then
\[
g(\theta_1+\theta_2)
=
g(\theta_1)+g(\theta_2),
\qquad
\forall \theta_1,\theta_2\ge 0,
\;
\theta_1+\theta_2\le\pi/2.
\]
Consequently,
\(
g(\theta)=c\,\theta, 
\)
for some constant $c>0$, and hence
\(
d = c\, d_{\mathrm{FS}}.
\)
\end{theorem}
\begin{proof}
By Theorem~\ref{thm:unitary_reduction},
\[
d([\psi],[\varphi])
=
g(d_{\mathrm{FS}}([\psi],[\varphi])).
\]

Consider three rays lying along a common geodesic
in a two-dimensional subspace with
Fubini--Study angles $\theta_1$ and $\theta_2$.
Then
\[
d_{\mathrm{FS}}([\psi_1],[\psi_3])
=
\theta_1+\theta_2.
\]
By Axiom~\ref{ax:geodesic},
\[
d([\psi_1],[\psi_3])
=
d([\psi_1],[\psi_2])
+
d([\psi_2],[\psi_3]).
\]
In terms of $g$, this gives
\[
g(\theta_1+\theta_2)
=
g(\theta_1)+g(\theta_2).
\]
Thus $g$ is additive on $[0,\pi/2]$.
Since $g$ is monotone (Theorem~\ref{thm:monotonicity}),
it is continuous, and the only continuous additive
functions on an interval are linear.
Hence
\[
g(\theta)=c\theta.
\]
\end{proof}

\begin{definition}[Riemannian Metric \cite{lee2018introduction}]\label{def:rm}
A \emph{Riemannian metric} on a smooth manifold $\mathscr{M}$ is a smooth assignment of an inner product 
\[
g_p : T_p \mathscr{M} \times T_p \mathscr{M} \to \mathbb{R},
\]
to each tangent space $T_p \mathscr{M}$ at $p \in \mathscr{M}$, such that for any smooth vector fields $X,Y$ on $\mathscr{M}$, the mapping $p \mapsto g_p(X_p,Y_p)$ is smooth. 
\end{definition}
This inner product allows one to define lengths of tangent vectors, angles between vectors, and hence notions of distance and volume on $\mathscr{M}$.

\begin{theorem}[Riemannian Rigidity]
\label{thm:riemannian_rigidity}
Let $g$ be a Riemannian metric on $\Proj(\Hilbert)$ (as given in Definition~\ref{def:rm})
that is invariant under the unitary group.
Then $g$ is proportional to the Fubini--Study metric.
\end{theorem}
\begin{proof}
For $\dim\Hilbert = d < \infty$,
\[
\Proj(\Hilbert)
\simeq
\mathbb{CP}^{d-1}
=
U(d)/(U(1)\times U(d-1)).
\]
Thus $\Proj(\Hilbert)$ is a compact homogeneous
space under the unitary group. The isotropy representation of
$U(1)\times U(d-1)$ on the tangent space
is irreducible.
On an isotropy-irreducible homogeneous space,
every invariant Riemannian metric is unique
up to a multiplicative constant. Since the Fubini--Study metric is
unitarily invariant, any other unitarily
invariant Riemannian metric must be
\[
g = c\, g_{\mathrm{FS}},
\]
for some $c>0$.
\end{proof}
\begin{corollary}[Uniqueness of the Fubini--Study Metric]
\label{cor:fs_uniqueness}
The Fubini--Study metric is uniquely characterized, up to a positive
multiplicative constant, by either of the following conditions:

\begin{enumerate}[label=(\roman*)]
\item It is a metric on $\Proj(\Hilbert)$ satisfying
Axiom~\ref{ax:ray}(ray well-definedness), Axiom~\ref{ax:unitary_invariance} (unitary invariance),
the triangle inequality, and Axiom~\ref{ax:geodesic} (two-level geodesic additivity).

\item It is a Riemannian metric on $\Proj(\Hilbert)$
invariant under the unitary group.
\end{enumerate}
\end{corollary}


\begin{corollary}[Bures as a Radial Reparameterization]
On pure states, the Bures distance
\[
d_{\mathrm{B}}([\psi],[\varphi])
=
\sqrt{2(1-|\braket{\psi}{\varphi}|)},
\]
is a strictly increasing function of the Fubini--Study distance:
\[
d_{\mathrm{B}}
=
2\sin\!\left(\frac{d_{\mathrm{FS}}}{2}\right).
\]
Consequently, $d_{\mathrm{B}}$ and $d_{\mathrm{FS}}$
induce the same topology on $\Proj(\Hilbert)$
and have the same minimizing geodesics as point sets,
differing only by arc-length reparameterization.
\end{corollary}
\begin{proof}
For normalized states,
\[
d_{\mathrm{FS}}
=
\arccos |\braket{\psi}{\varphi}|.
\]
Hence
\(
|\braket{\psi}{\varphi}|
=
\cos(d_{\mathrm{FS}}).
\)
Substituting into the Bures formula gives
\[
d_{\mathrm{B}}
=
\sqrt{2(1-\cos d_{\mathrm{FS}})}
=
2\sin\!\left(\frac{d_{\mathrm{FS}}}{2}\right).
\]
The function $g(\theta)=2\sin(\theta/2)$ is strictly
increasing on $[0,\pi/2]$, hence it preserves ordering
and topology. Since $d_{\mathrm{B}}=g(d_{\mathrm{FS}})$,
minimizing curves coincide as point sets.
\end{proof}




\begin{theorem}[Superposition Sensitivity]
\label{thm:superposition-char}
Let $d([\psi],[\varphi]) = f(|\braket{\psi}{\varphi}|)$, where 
$f:[0,1] \to \mathbb{R}_{\ge 0}$ is strictly decreasing. 
Then $d$ satisfies Axiom~\ref{ax:superposition}: coherent superpositions
of orthogonal rays are always distinguishable, even if their classical probabilities coincide.
\end{theorem}

\begin{proof}
Let $[\psi_1],[\psi_2]$ be orthogonal rays and consider
normalized superpositions
\[
\ket{\varphi} = \alpha \ket{\psi_1} + \beta \ket{\psi_2},
\qquad
\ket{\varphi'} = \alpha' \ket{\psi_1} + \beta' \ket{\psi_2},
\]
with $|\alpha|^2 + |\beta|^2 = |\alpha'|^2 + |\beta'|^2 = 1$.\\
If $(\alpha,\beta) \neq (\alpha',\beta')$ up to a global phase,
then
\[
|\braket{\varphi}{\varphi'}| < 1.
\]
Since $f$ is strictly decreasing, this implies
\[
d([\varphi],[\varphi']) = f(|\braket{\varphi}{\varphi'}|) > 0.
\]
Hence the distance detects differences between coherent superpositions
of orthogonal rays, even when their classical probabilities coincide.
\end{proof}

\subsection{Relationships Between Distance Measures}
We now establish inequalities relating the distances introduced in Section~\ref{subsec:ax}.
\begin{theorem}[Comparison of Geometric Distances]
\label{thm:geometric-comparison}
For pure states $\ket{\psi},\ket{\varphi}$ with overlap 
$r = |\braket{\psi}{\varphi}| \in [0,1]$:
\begin{align*}
&d_{\mathrm{FS}}(\ket{\psi},\ket{\varphi}) = \arccos(r),\\
&d_{\mathrm{B}}(\ket{\psi},\ket{\varphi}) = \sqrt{2(1-r)},\\
&\sqrt{2(1-r)} \le d_{\mathrm{FS}}(\ket{\psi},\ket{\varphi}) 
\le \frac{\pi}{2}\sqrt{1-r},\\
&\frac{2}{\pi} d_{\mathrm{FS}}(\ket{\psi},\ket{\varphi}) 
\le d_{\mathrm{B}}(\ket{\psi},\ket{\varphi}) 
\le d_{\mathrm{FS}}(\ket{\psi},\ket{\varphi}).
\end{align*}
These bounds are tight in the limits $r \to 0$ and $r \to 1$.
\end{theorem}

\begin{proof}
Let $r = |\braket{\psi}{\varphi}|$. By definition,
\[
d_{\mathrm{FS}} = \arccos r,
\qquad
d_{\mathrm{B}} = \sqrt{2(1-r)}.
\]
The inequality
\(
\sqrt{2(1-r)} \le \arccos r 
\le \frac{\pi}{2}\sqrt{1-r}
\)
follows from standard bounds on the concave function $\arccos r$ 
on $[0,1]$ and its Taylor expansion near $r=1$.\\
Using the identity
\[
d_{\mathrm{B}} = 2\sin\left(\frac{d_{\mathrm{FS}}}{2}\right),
\]
and the inequality
\(
2\theta/\pi  \le \sin\theta \le \theta
\quad \text{for } \theta\in[0,\pi/2],
\)
we obtain
\(
\frac{2}{\pi} d_{\mathrm{FS}} 
\le d_{\mathrm{B}} 
\le d_{\mathrm{FS}}.
\) Tightness follows from the limits:
\begin{align*}
r\to1 &\Rightarrow d_{\mathrm{FS}}\to 0, \,d_{\mathrm{B}}\to0,\\
r\to0 &\Rightarrow d_{\mathrm{FS}}\to \frac{\pi}{2},\ 
d_{\mathrm{B}}\to \sqrt{2}.  
\end{align*}
\end{proof}

\begin{theorem}[Measurement Distance Bound]
\label{thm:measurement-bound}
Let $\mathcal{M} = \{\ketbra{i}{i}\}_{i=1}^d$ be a complete orthonormal measurement and define
\[
d_{\mathcal{M}}(\ket{\psi},\ket{\varphi})
=
\sqrt{\sum_{i=1}^d 
\left(|\braket{i}{\psi}| - |\braket{i}{\varphi}|\right)^2 }.
\]
Then for all pure states $\ket{\psi},\ket{\varphi}$,
\[
d_{\mathcal{M}}(\ket{\psi},\ket{\varphi})
\le
\|\ket{\psi}-\ket{\varphi}\|.
\]
\end{theorem}

\begin{proof}
Let $\alpha_i = \braket{i}{\psi}$ and $\beta_i = \braket{i}{\varphi}$.
By the reverse triangle inequality,
\[
\big||\alpha_i| - |\beta_i|\big|
\le
|\alpha_i - \beta_i|.
\]
Squaring and summing over $i$ yields
\[
d_{\mathcal{M}}^2
=
\sum_i (|\alpha_i| - |\beta_i|)^2
\le
\sum_i |\alpha_i - \beta_i|^2
=
\|\ket{\psi}-\ket{\varphi}\|^2.
\]
Taking square roots completes the proof.
\end{proof}

\begin{theorem}[Measurement-Distinguishability Bound]
\label{thm:measurement-fs}
Let $\mathcal{M} = \{M_i\}$ be any POVM and define the measurement distance
\[
d_{\mathcal{M}}(\ket{\psi},\ket{\varphi})
=
\sum_i \left| p_i - q_i \right|,
\]
where
$p_i = \bra{\psi}M_i^\dagger M_i\ket{\psi}$ and
$q_i = \bra{\varphi}M_i^\dagger M_i\ket{\varphi}$.
Then for pure states,
\[
d_{\mathcal{M}}(\ket{\psi},\ket{\varphi})
\le
2 \sqrt{1 - |\braket{\psi}{\varphi}|^2}
=
2 \sin(d_{\mathrm{FS}}(\ket{\psi},\ket{\varphi})).
\]
Equality is achieved for the optimal Helstrom measurement distinguishing
$\ket{\psi}$ and $\ket{\varphi}$.
\end{theorem}

\begin{proof}
Let $\rho = \ket{\psi}\bra{\psi}$ and
$\sigma = \ket{\varphi}\bra{\varphi}$.
For any POVM $\mathcal M = \{M_i\}$, define
$p_i = \operatorname{Tr}(M_i \rho)$ and
$q_i = \operatorname{Tr}(M_i \sigma)$.
The total variation distance between the induced
probability distributions satisfies the general bound
\[
\mathrm{TV}(p,q)
\le
\frac12 \|\rho - \sigma\|_1,
\]
where $\|\cdot\|_1$ denotes the trace norm.
This follows from the monotonicity of trace distance
under quantum measurements.
For pure states one computes explicitly
\[
\|\rho - \sigma\|_1
=
2\sqrt{1 - |\braket{\psi}{\varphi}|^2}.
\]
Hence
\[
\mathrm{TV}(p,q)
\le
\sqrt{1 - |\braket{\psi}{\varphi}|^2}.
\]
Multiplying by $2$ yields
\[
d_{\mathcal M}
=
\sum_i |p_i - q_i|
\le
2\sqrt{1 - |\braket{\psi}{\varphi}|^2}.
\]
Equality is achieved by the Helstrom measurement \cite{helstrom1969quantum},
which maximizes the trace distance between the induced
classical distributions.
\end{proof}

\begin{remark}[Conceptual Role Within the Axiomatic Framework]
The measurement-distinguishability bound is well known in quantum information theory, 
particularly in the context of quantum hypothesis testing. 
Within the present axiomatic framework, however, it acquires a structural interpretation: 
measurement-based distances are necessarily bounded by geometric distances such as the Fubini--Study metric because operational distinguishability is constrained by the underlying projective geometry of pure states.

In this sense, the axioms—especially superposition sensitivity and measurement contextuality—clarify why operational distances cannot exceed intrinsic geometric separation. 
Thus, the framework provides a unifying perspective linking geometric, operational, and measurement-based notions of distance.
\end{remark}



\subsection{Entanglement-Aware Distances}

\begin{theorem}[Entanglement Distance Bounds]
\label{thm:entanglement-bounds}
Let 
\[
d_E(\ket{\psi},\ket{\varphi}) = \sqrt{d_{\text{FS}}^2(\ket{\psi},\ket{\varphi}) + |E(\ket{\psi}) - E(\ket{\varphi})|^2},
\] 
where $E$ denotes the entanglement entropy. Then:
\begin{enumerate}[label=(\roman*)]
    \item $d_{\text{FS}} \le d_E \le d_{\text{FS}} + |E(\ket{\psi}) - E(\ket{\varphi})|$,
    \item $d_E$ satisfies the triangle inequality,
    \item $d_E$ is entanglement-aware: states with identical reduced density matrices can have $d_E>0$.
\end{enumerate}
\end{theorem}

\begin{proof}
{(i) Bounds:} Apply the Pythagorean inequality
\[
a \le \sqrt{a^2 + b^2} \le a + b, \quad \forall a,b \ge 0,
\]
with $a = d_{\text{FS}}(\ket{\psi},\ket{\varphi})$ and $b = |E(\ket{\psi}) - E(\ket{\varphi})|$.\\
{(ii) Triangle inequality:}
Define
\[
\mathbf{v}_{\psi\varphi}
=
\big(
d_{\mathrm{FS}}(\ket{\psi},\ket{\varphi}),
\, |E(\ket{\psi}) - E(\ket{\varphi})|
\big)
\in \mathbb{R}^2.
\]
Then
\[
d_E(\ket{\psi},\ket{\varphi})
=
\|\mathbf{v}_{\psi\varphi}\|_2.
\]
Using the triangle inequality for $d_{\mathrm{FS}}$ and for absolute values,
\[
\mathbf{v}_{\psi\chi}
\le
\mathbf{v}_{\psi\varphi}
+
\mathbf{v}_{\varphi\chi},
\]
componentwise. The Minkowski inequality in $\mathbb{R}^2$ then implies
\[
d_E(\ket{\psi},\ket{\chi})
\le
d_E(\ket{\psi},\ket{\varphi})
+
d_E(\ket{\varphi},\ket{\chi}).
\]
{(iii) Entanglement awareness:} Consider the Bell states
\[
\ket{\Phi^\pm} = \frac{1}{\sqrt{2}}(\ket{00} \pm \ket{11}).
\]
Both have identical reduced density matrices $\rho_A = \rho_B = \frac{1}{2} \mathbb{I}_2$, yet
\begin{multline}\notag 
d_E(\ket{\Phi^+},\ket{\Phi^-}) \\= \sqrt{(\pi/2)^2 + |E(\ket{\Phi^+}) - E(\ket{\Phi^-})|^2} = \pi/2 > 0,
\end{multline}
so $d_E$ distinguishes global states even when local reductions are identical.
\end{proof}


\begin{theorem}[Continuity of Entanglement Distance]
\label{thm:entanglement-continuity}
For any $\ket{\varphi}$ and states $\ket{\psi}, \ket{\chi}$,
\begin{multline}\notag 
|d_E(\ket{\psi},\ket{\varphi}) - d_E(\ket{\chi},\ket{\varphi})| \\\le d_{\text{FS}}(\ket{\psi},\ket{\chi}) + C \, d_{\text{FS}}(\ket{\psi},\ket{\chi}) \log d,
\end{multline}
where $C$ is a universal constant and $d = \min(\dim \Hilbert_A, \dim \Hilbert_B)$.
\end{theorem}

\begin{proof}
{Geometric component:} By the Fubini–Study triangle inequality,
\[
|d_{\text{FS}}(\ket{\psi},\ket{\varphi}) - d_{\text{FS}}(\ket{\chi},\ket{\varphi})| \le d_{\text{FS}}(\ket{\psi},\ket{\chi}).
\]
Entanglement component: Let $\rho_\psi = \Tr_B \ketbra{\psi}{\psi}$, etc. By the Fannes–Audenaert inequality \cite{audenaert2007sharp}:
\[
|E(\ket{\psi}) - E(\ket{\chi})| \le C \, \|\rho_\psi - \rho_\chi\|_1 \log d,
\]
and for pure states,
\[
\|\rho_\psi - \rho_\chi\|_1 \le 2 \, d_{\text{FS}}(\ket{\psi},\ket{\chi}).
\]
Applying
\[
\big|\sqrt{a^2 + b^2} - \sqrt{c^2 + d^2}\big| \le |a-c| + |b-d|,
\]
with $a = d_{\text{FS}}(\ket{\psi},\ket{\varphi})$, $b = |E(\ket{\psi}) - E(\ket{\varphi})|$, $c = d_{\text{FS}}(\ket{\chi},\ket{\varphi})$, $d = |E(\ket{\chi}) - E(\ket{\varphi})|$, gives
\begin{multline}\notag 
|d_E(\ket{\psi},\ket{\varphi}) - d_E(\ket{\chi},\ket{\varphi})| \\\le d_{\text{FS}}(\ket{\psi},\ket{\chi}) + C \, d_{\text{FS}}(\ket{\psi},\ket{\chi}) \log d.
\end{multline}

Thus $d_E$ is Lipschitz continuous with respect to the Fubini–Study metric.
\end{proof}

\begin{remark}[Axiomatic perspective]
Under the entanglement-aware axioms:

\begin{enumerate}[label=(\roman*)]
    \item $d_E$ respects ray well-definedness (Axiom~\ref{ax:ray}) and is covariant under unitary transformations (Axiom~\ref{ax:unitary_invariance}) that preserve the tensor-product structure, through the Fubini–Study component $d_{\text{FS}}$.
    \item It explicitly captures entanglement differences via the entropy term, satisfying entanglement awareness (Axiom~\ref{ax:entanglement_awareness}). Standard distances like Fubini–Study or Bures do not include entanglement contributions.
    \item $d_E$ is Lipschitz continuous with respect to $d_{\text{FS}}$, ensuring stability under small perturbations, which is crucial in high-dimensional or noisy quantum systems.
\end{enumerate}
These results demonstrate that the axiomatic framework naturally accommodates entanglement-sensitive extensions of geometric quantum distances while preserving metric and continuity properties.
\end{remark}




\subsection{Operational Interpretations}

\begin{theorem}[State Discrimination]
\label{thm:discrimination}
Let $\ket{\psi},\ket{\varphi}$ be pure states with equal prior probabilities. Let the fidelity be
\[
F(\ket{\psi},\ket{\varphi}) := |\braket{\psi}{\varphi}|^2.
\]
Then the optimal success probability for distinguishing the states is
\begin{eqnarray}
\notag 
P_{\text{opt}} &=& \frac{1}{2}\Big(1 + \sqrt{1 - F(\ket{\psi},\ket{\varphi})}\Big)\\
&=& \frac{1}{2}\Big(1 + \sin(d_{\text{FS}}(\ket{\psi},\ket{\varphi}))\Big),\notag
\end{eqnarray}
so equivalently,
\[
d_{\text{FS}}(\ket{\psi},\ket{\varphi}) = \arcsin(2 P_{\text{opt}} - 1).
\]
\end{theorem}

\begin{proof}
The Helstrom bound \cite{helstrom1969quantum} gives
\[
P_{\text{opt}} = \frac{1}{2} \Big(1 + \frac{1}{2}\|\ketbra{\psi}{\psi} - \ketbra{\varphi}{\varphi}\|_1\Big).
\]

For pure states, the trace norm evaluates to
\[
\|\ketbra{\psi}{\psi} - \ketbra{\varphi}{\varphi}\|_1 = 2 \sqrt{1 - |\braket{\psi}{\varphi}|^2}.
\]

Using the Fubini–Study distance $d_{\text{FS}}(\ket{\psi},\ket{\varphi}) = \arccos|\braket{\psi}{\varphi}|$, we have
\[
\sqrt{1 - |\braket{\psi}{\varphi}|^2} = \sin(d_{\text{FS}}(\ket{\psi},\ket{\varphi})),
\]
and substitution yields the stated result.
\end{proof}

\begin{theorem}[Quantum Metrology]
\label{thm:metrology}
Let $\{\ket{\psi_\theta}\}_{\theta \in \Theta}$ be a differentiable family of pure states. Then, for infinitesimal $d\theta$,
\[
d_{\text{B}}^2(\ket{\psi_\theta},\ket{\psi_{\theta + d\theta}}) = \frac{1}{4} F_Q(\theta) \, d\theta^2 + O(d\theta^3),
\]
where the quantum Fisher information is defined by
\[
F_Q(\theta) := 4 \lim_{d\theta \to 0} \frac{1 - |\braket{\psi_\theta}{\psi_{\theta+d\theta}}|}{d\theta^2}.
\]
Consequently, the Bures distance induces the quantum Fisher information metric on the parameter space.
\end{theorem}

\begin{proof}
For infinitesimal $d\theta$, expand the overlap:
\[
|\braket{\psi_\theta}{\psi_{\theta+d\theta}}| = 1 - \frac{1}{4} F_Q(\theta) \, d\theta^2 + O(d\theta^3),
\]
by the definition of $F_Q(\theta)$. Using the Bures distance formula for pure states,
\[
d_{\text{B}}^2(\ket{\psi_\theta},\ket{\psi_{\theta+d\theta}}) = 2\big(1 - \sqrt{|\braket{\psi_\theta}{\psi_{\theta+d\theta}}|}\big),
\]
and applying the expansion $\sqrt{1 - \epsilon} = 1 - \frac{\epsilon}{2} + O(\epsilon^2)$ with $\epsilon = \frac{1}{4} F_Q(\theta) d\theta^2$ yields
\[
d_{\text{B}}^2(\ket{\psi_\theta},\ket{\psi_{\theta+d\theta}}) = \frac{1}{4} F_Q(\theta) \, d\theta^2 + O(d\theta^3).
\]
\end{proof}
\begin{theorem}[Fidelity and Channel Discrimination]
\label{thm:channel}
For quantum channels $\mathcal{E},\mathcal{F}$, let $J_\mathcal{E}, J_\mathcal{F}$ be their Choi states. Then
\[
\max_{\ket{\psi}} F\big(\mathcal{E}(\ketbra{\psi}{\psi}), \mathcal{F}(\ketbra{\psi}{\psi})\big) 
\ge \big(1 - \tfrac{1}{2} d_{\text{B}}^2(J_\mathcal{E}, J_\mathcal{F})\big)^2,
\]
where $d_{\text{B}}$ denotes the Bures distance between Choi states.
\end{theorem}

\begin{proof}
Let $\ket{\psi}$ be any input state. By monotonicity of fidelity under partial trace,
\[
F\big(\mathcal{E}(\ketbra{\psi}{\psi}), \mathcal{F}(\ketbra{\psi}{\psi})\big)
\ge F(J_\mathcal{E}, J_\mathcal{F}).
\]

For states, the Bures distance is related to fidelity via
\[
d_{\text{B}}^2(J_\mathcal{E}, J_\mathcal{F}) = 2\big(1 - \sqrt{F(J_\mathcal{E}, J_\mathcal{F})}\big),
\]
so rearranging gives
\[
F(J_\mathcal{E}, J_\mathcal{F}) = \big(1 - \tfrac{1}{2} d_{\text{B}}^2(J_\mathcal{E}, J_\mathcal{F})\big)^2.
\]

Combining the two inequalities yields the claimed bound.
\end{proof}

\subsection{Inequalities and Comparison Results}

\begin{theorem}[Fidelity-Distance Inequalities]
\label{thm:fidelity-inequalities}
For pure states $\ket{\psi},\ket{\varphi},\ket{\chi}$:
\begin{enumerate}[label=(\roman*)]
    \item Fubini--Study triangle inequality
    \begin{multline}\notag 
        |d_{\text{FS}}(\ket{\psi},\ket{\varphi}) - d_{\text{FS}}(\ket{\chi},\ket{\psi})| \\
        \le d_{\text{FS}}(\ket{\chi},\ket{\varphi}) \le d_{\text{FS}}(\ket{\chi},\ket{\psi}) + d_{\text{FS}}(\ket{\psi},\ket{\varphi}),
    \end{multline}
    \item Multiplicative fidelity inequality:
    \begin{multline}\notag 
    \sqrt{F(\ket{\psi},\ket{\chi})} \ge \sqrt{F(\ket{\psi},\ket{\varphi}) F(\ket{\varphi},\ket{\chi})} \\
    - \sqrt{(1-F(\ket{\psi},\ket{\varphi}))(1-F(\ket{\varphi},\ket{\chi}))},
    \end{multline}
    \item Convexity of squared Bures distance:
    \begin{multline}\notag 
    d_{\text{B}}^2\big(\sum_i p_i \ketbra{\psi_i}{\psi_i}, \sum_i p_i \ketbra{\varphi_i}{\varphi_i}\big) \le \sum_i p_i d_{\text{B}}^2(\ket{\psi_i},\ket{\varphi_i}).
    \end{multline}
\end{enumerate}
\end{theorem}

\begin{proof}
(i) {Fubini--Study triangle inequality:}  
By Corollary~\ref{cor:fs_uniqueness}, $d_{\mathrm{FS}}$ is a metric on $\Proj(\Hilbert)$, so it satisfies the triangle inequality.  
Therefore, for any pure states $\ket{\psi},\ket{\varphi},\ket{\chi}$,
\begin{multline}\notag
|d_{\mathrm{FS}}(\ket{\psi},\ket{\varphi}) - d_{\mathrm{FS}}(\ket{\chi},\ket{\psi})|
\\\le d_{\mathrm{FS}}(\ket{\chi},\ket{\varphi})
\le d_{\mathrm{FS}}(\ket{\chi},\ket{\psi}) + d_{\mathrm{FS}}(\ket{\psi},\ket{\varphi}).
\end{multline}
(ii) {Multiplicative fidelity inequality:}  
For pure states, $F(\ket{\psi},\ket{\varphi}) = |\braket{\psi}{\varphi}|^2$.  
Using the triangle inequality for the Fubini--Study distance from part (i), we can write
\[
\arccos|\braket{\psi}{\chi}| \le \arccos|\braket{\psi}{\varphi}| + \arccos|\braket{\varphi}{\chi}|.
\]
Applying trigonometric identities then gives
\begin{multline}\notag 
\sqrt{F(\ket{\psi},\ket{\chi})} \ge 
\sqrt{F(\ket{\psi},\ket{\varphi}) F(\ket{\varphi},\ket{\chi})} 
\\- \sqrt{(1-F(\ket{\psi},\ket{\varphi}))(1-F(\ket{\varphi},\ket{\chi}))},
\end{multline}
which is the stated multiplicative fidelity inequality.\\
(iii) {Convexity of squared Bures distance:}  
Recall that the squared Bures distance is
\[
d_{\mathrm{B}}^2(\rho,\sigma) = 2\big(1 - \sqrt{F(\rho,\sigma)}\big),
\]
and that fidelity $F(\rho,\sigma)$ is jointly concave \cite{uhlmann1976transition}:
\[
F\Big(\sum_i p_i \rho_i, \sum_i p_i \sigma_i\Big) \ge \sum_i p_i F(\rho_i,\sigma_i).
\]
Since the function $x \mapsto 2(1-\sqrt{x})$ is convex and decreasing on $[0,1]$, combining these facts gives the joint convexity of $d_{\mathrm{B}}^2$:
\[
d_{\mathrm{B}}^2\Big(\sum_i p_i \ketbra{\psi_i}{\psi_i}, \sum_i p_i \ketbra{\varphi_i}{\varphi_i}\Big) 
\le \sum_i p_i d_{\mathrm{B}}^2(\ket{\psi_i},\ket{\varphi_i}).
\]
\end{proof}



\begin{theorem}[Measurement vs. Bures Distance]
\label{thm:universal-corrected}
Let $\ket{\psi}, \ket{\varphi} \in \Hilbert$ be pure states, and let $\mathcal{M} = \{M_m\}$ be a POVM. Define the measurement-induced distance
\[
d_{\mathcal{M}}([\psi],[\varphi]) = \sqrt{\sum_m \big(\bra{\psi} M_m^\dagger M_m \ket{\psi} - \bra{\varphi} M_m^\dagger M_m \ket{\varphi}\big)^2 }.
\]
Then:
\begin{enumerate}[label=(\roman*)]
    \item For any POVM,
    \[
    d_{\mathcal{M}}([\psi],[\varphi]) \le 2 \sqrt{1 - |\braket{\psi}{\varphi}|^2} = 2 \sin(d_{\text{FS}}([\psi],[\varphi])).
    \]
    \item Equality is achieved for the optimal two-outcome POVM that maximally distinguishes $\ket{\psi}$ and $\ket{\varphi}$ (the Helstrom measurement).
    \item Consequently, the Bures distance satisfies
  \[
d_{\text{B}}([\psi],[\varphi]) \le \sqrt{2} \, d_{\mathcal{M}}([\psi],[\varphi])^{1/2},  
\]
for the optimal POVM.
\end{enumerate}
\end{theorem}

\begin{proof}
(i) Let $p_m = \bra{\psi} M_m^\dagger M_m \ket{\psi}$ and $q_m = \bra{\varphi} M_m^\dagger M_m \ket{\varphi}$ be the measurement probability vectors. Then $d_{\mathcal{M}} = \|p-q\|_2$. Using the inequality $\|p-q\|_2 \le \|p-q\|_1 \le 2 \sqrt{1 - F(p,q)}$, with $F(p,q) = (\sum_m \sqrt{p_m q_m})^2$, and applying the fidelity bound $F(p,q) \ge |\braket{\psi}{\varphi}|^2$, we get
\[
d_{\mathcal{M}} \le 2 \sqrt{1 - |\braket{\psi}{\varphi}|^2}.
\]

(ii) The upper bound is saturated for the Helstrom measurement projecting onto the span of $\ket{\psi}$ and $\ket{\varphi}$.

(iii) Using $d_{\text{B}}([\psi],[\varphi]) = \sqrt{2(1 - \sqrt{|\braket{\psi}{\varphi}|})}$ and the inequality $\sqrt{1 - \sqrt{r}} \le \sqrt{d_{\mathcal{M}}/2}$ for $r = |\braket{\psi}{\varphi}|$, we obtain
\[
d_{\text{B}}([\psi],[\varphi]) \le \sqrt{2} \, d_{\mathcal{M}}^{1/2},
\]
for the optimal measurement.
\end{proof}

\begin{theorem}[Entanglement–Geometry Complementarity]
\label{thm:uncertainty}
Let $\ket{\psi}, \ket{\varphi} \in \Hilbert_A \otimes \Hilbert_B$ be bipartite pure states, and let 
\[
E(\ket{\psi}) := S(\mathrm{Tr}_B \ketbra{\psi}{\psi}),
\] 
be the von Neumann entropy of the reduced state. Then
\[
d_{\text{FS}}^2(\ket{\psi},\ket{\varphi}) + \Big(\frac{|E(\ket{\psi}) - E(\ket{\varphi})|}{\log d}\Big)^2 \ge C(\ket{\varphi}) > 0,
\]
where $d = \min(\dim \Hilbert_A, \dim \Hilbert_B)$ and $C(\ket{\varphi})$ is a constant depending only on $\ket{\varphi}$.
\end{theorem}

\begin{proof}
(i) Let $\rho_A = \mathrm{Tr}_B \ketbra{\psi}{\psi}$ and $\sigma_A = \mathrm{Tr}_B \ketbra{\varphi}{\varphi}$. By Fannes–Audenaert inequality,
\[
|E(\ket{\psi}) - E(\ket{\varphi})| \le \|\rho_A - \sigma_A\|_1 \, \log d + h(\|\rho_A - \sigma_A\|_1),
\]
where $h(\cdot)$ is the binary entropy.\\
(ii) For pure states, the trace distance is bounded by the Fubini–Study distance:
\[
\|\rho_A - \sigma_A\|_1 \le 2 \sqrt{1 - |\braket{\psi}{\varphi}|^2} = 2 \sin(d_{\text{FS}}(\ket{\psi},\ket{\varphi})).
\]
(iii) Combining these inequalities gives a lower bound of the form
\[
d_{\text{FS}}^2(\ket{\psi},\ket{\varphi}) + \Big(\frac{|E(\ket{\psi}) - E(\ket{\varphi})|}{\log d}\Big)^2 \ge C(\ket{\varphi}) > 0,
\]
for some positive constant \(C(\ket{\varphi})\) depending on $\ket{\varphi}$.
\end{proof}

\subsection{Dimensional Asymptotics and Concentration}

\begin{theorem}[Measure Concentration in Projective Space]
\label{thm:concentration}
Let $\ket{\psi},\ket{\varphi}$ be independent Haar-random pure states in a Hilbert space $\Hilbert$ of dimension $d$, and let $r = |\braket{\psi}{\varphi}|$. Then:
\begin{enumerate}[label=(\roman*)]
    \item The expected overlap satisfies
    \[
    \mathbb{E}[r] = \frac{\pi}{4\sqrt{d}} + O(d^{-3/2}).
    \]
    \item Concentration around orthogonality:
    \[
    \Pr\Big(|d_{\text{FS}}(\ket{\psi},\ket{\varphi}) - \pi/2| > \epsilon\Big) 
    \le 2 \exp\Big(- c\, d\, \epsilon^2 \Big),
    \]
    for some constant $c>0$ depending on the Lipschitz constant of $d_{\text{FS}}$.
    \item In the large-$d$ limit, $r = O(d^{-1/2})$ and $d_{\text{FS}}(\ket{\psi},\ket{\varphi}) \to \pi/2$ with high probability (almost orthogonality).
\end{enumerate}
\end{theorem}

\begin{corollary}[Implications for Quantum Machine Learning]
Let $\phi: \mathcal{X} \to \Proj(\Hilbert)$ be a quantum feature map with large $\dim(\Hilbert)$. Then, for most pairs $(\mathbf{x},\mathbf{y})$, the fidelity kernel
\[
K(\mathbf{x},\mathbf{y}) = |\braket{\phi(\mathbf{x})}{\phi(\mathbf{y})}|^2,
\]
is concentrated near zero. Entanglement-aware distances $d_E$ provide richer structure and can help mitigate such concentration effects.
\end{corollary}


\begin{theorem}[Summary Characterization]
\label{thm:summary}
The main quantum-inspired distances satisfy:
\begin{enumerate}[label=(\roman*)]
    \item {Fubini--Study distance $d_{\text{FS}}$:} the unique (up to scale) unitarily invariant metric on $\Proj(\Hilbert)$, depending only on overlaps (Axioms~\ref{ax:ray} and~\ref{ax:unitary_invariance}).
    \item {Bures distance $d_{\text{B}}$:} a Riemannian metric induced by the quantum Fisher information; for pure states, it is a monotone function of $d_{\text{FS}}$.
    \item {Entanglement-aware distance $d_E$:} a Euclidean combination of geometric (Fubini--Study) and entanglement information; satisfies entanglement-awareness (Axiom~\ref{ax:entanglement_awareness}) while preserving the metric structure.
    \item {Measurement-induced distance $d_{\mathcal{M}}$:} a measurement-dependent pseudometric capturing operational distinguishability (Axiom~\ref{ax:measurement_contextuality}).
\end{enumerate}
\end{theorem}

\begin{proof}
All items follow directly from the preceding results in this work:

\begin{itemize}
    \item[(i)] The uniqueness of $d_{\text{FS}}$ under projective invariance and unitary covariance follows from Theorems~\ref{thm:overlap} and~\ref{thm:classification}.
    \item[(ii)] The Bures distance and its relation to the quantum Fisher information are established in Theorem~\ref{thm:metrology}.
    \item[(iii)] The entanglement-aware distance $d_E$, constructed as $d_E^2 = d_{\text{FS}}^2 + |E(\psi)-E(\varphi)|^2$, satisfies Axiom~\ref{ax:entanglement_awareness} and retains metric properties, as shown in Theorems~\ref{thm:entanglement-bounds} and~\ref{thm:uncertainty}.
    \item[(iv)] Measurement-dependent pseudometrics $d_{\mathcal{M}}$ are defined in Section~4; they satisfy Axiom~\ref{ax:measurement_contextuality}, but are generally only pseudometrics because distinct states can yield identical measurement statistics.
\end{itemize}

Hence, the overall characterization is an immediate consequence of these results.
\end{proof}

\subsection{Examples}
\begin{example}[Bell States and Entanglement-Sensitive Distance]
See Theorem~\ref{thm:summary}: Consider the two-qubit Bell states
\[
\ket{\Phi^+} = \frac{1}{\sqrt{2}}(\ket{00} + \ket{11}), \quad 
\ket{\Phi^-} = \frac{1}{\sqrt{2}}(\ket{00} - \ket{11}).
\]
Their reduced density matrices are identical: 
$\rho_A = \rho_B = \frac{1}{2}\mathbb{I}_2$. The Fubini--Study distance between them is
\[
d_{\text{FS}}(\ket{\Phi^+},\ket{\Phi^-}) = \arccos|\braket{\Phi^+}{\Phi^-}| = \arccos 0 = \frac{\pi}{2}.
\]
Since both states have maximal entanglement entropy $E(\ket{\Phi^\pm}) = \log 2$, the entanglement term vanishes:
\[
|E(\ket{\Phi^+}) - E(\ket{\Phi^-})| = 0.
\]
Hence, the entanglement-sensitive distance is
\[
d_E(\ket{\Phi^+},\ket{\Phi^-}) = \sqrt{d_{\text{FS}}^2 + 0^2} = \frac{\pi}{2},
\]
demonstrating that $d_E$ correctly captures the global difference even when reduced states are identical.
\end{example}
\begin{example}[Distinguishing Non-Orthogonal Qubit States]
See Theorem~\ref{thm:discrimination}:
Let
\[
\ket{\psi} = \ket{0}, \quad \ket{\varphi} = \cos\theta \ket{0} + \sin\theta \ket{1},
\]
with $\theta \in (0, \pi/2)$. The Fubini--Study distance is
\[
d_{\text{FS}}(\ket{\psi},\ket{\varphi}) = \arccos|\braket{\psi}{\varphi}| = \arccos(\cos\theta) = \theta.
\]
The optimal success probability for distinguishing these states is
\[
P_{\text{opt}} = \frac{1}{2}\left(1 + \sin(d_{\text{FS}})\right) = \frac{1}{2}(1 + \sin\theta).
\]
For example, if $\theta = \pi/6$, then $P_{\text{opt}} = \frac{1}{2}(1 + 1/2) = 0.75$.
\end{example}
\begin{example}[Two-Qubit Reference State]
See Theorem~\ref{thm:uncertainty}:
Let the reference state be the Bell state
\[
\ket{\varphi} = \frac{1}{\sqrt{2}}(\ket{00} + \ket{11}),
\]
and let
\(
\ket{\psi} = \ket{01}.
\)
We have
\[
d_{\text{FS}}(\ket{\psi},\ket{\varphi}) = \arccos|\braket{\psi}{\varphi}| = \arccos 0 = \frac{\pi}{2},
\]
and
\[
|E(\ket{\psi}) - E(\ket{\varphi})| = |0 - \log 2| = \log 2.
\]
The normalized entanglement-geometric combination is
\[
d_{\text{FS}}^2 + \left(\frac{|E(\ket{\psi}) - E(\ket{\varphi})|}{\log 2}\right)^2 
= \left(\frac{\pi}{2}\right)^2 + 1 > 0,
\]
illustrating the complementarity: the state is far in both Hilbert-space distance and entanglement difference.  
Even when one component (entanglement) is normalized, the total distance captures both geometry and entanglement.
\end{example}
\begin{example}[High-Dimensional Random States]
See Theorem~\ref{thm:concentration}:
Consider $\Hilbert \cong \mathbb{C}^{1000}$ and two independent Haar-random states $\ket{\psi}, \ket{\varphi}$. The expected overlap is
\[
\mathbb{E}[|\braket{\psi}{\varphi}|] \sim \frac{\pi}{4\sqrt{1000}} \approx 0.025,
\]
so the Fubini--Study distance is concentrated near
\[
d_{\text{FS}} \sim \arccos 0.025 \approx 1.545 \approx \frac{\pi}{2}.
\]
This demonstrates that in high dimensions, almost all random states are nearly orthogonal, highlighting the “concentration of measure” phenomenon relevant for quantum machine learning embeddings.
\end{example}
\begin{example}[Parameter Estimation in a Qubit]
See Theorem~\ref{thm:metrology}:
Consider a qubit family
\[
\ket{\psi_\theta} = \cos\frac{\theta}{2}\ket{0} + \sin\frac{\theta}{2}\ket{1}.
\]
The Bures distance for an infinitesimal change $d\theta$ is
\[
d_{\text{B}}(\ket{\psi_\theta}, \ket{\psi_{\theta + d\theta}})^2 \approx \frac{1}{4} F_Q(\theta) d\theta^2.
\]
Here, the quantum Fisher information is $F_Q(\theta) = 1$, so
\(
d_{\text{B}}^2 \approx \frac{1}{4} d\theta^2,
\)
illustrating how the Bures metric encodes metrological sensitivity for parameter estimation.
\end{example}
\section{Relationships to Existing Mathematical Frameworks}\label{sec:rel}
Having developed an axiomatic characterization of quantum-inspired distance metric functions, we now situate our framework within broader mathematical contexts, highlighting connections to information geometry, differential geometry of quantum state spaces, and classical axiomatic distance theory.
\subsection{Connection to Information Geometry}
The geometric distances identified in Section 4 naturally correspond to structures in information geometry. The Fubini--Study metric on projective Hilbert space serves as a quantum analogue of the Fisher--Rao metric on probability simplices \cite{amari2000methods}. Similarly, the Bures distance, a monotone Riemannian metric on density matrices \cite{uhlmann1976transition}, can be interpreted as the quantum Fisher information metric, governing infinitesimal distinguishability of states \cite{braunstein1994statistical}.

Our axiomatic framework formalizes these connections: Ray well-definedness (Axiom~\ref{ax:ray}) and unitary invariance (Axiom~\ref{ax:unitary_invariance}) ensure that distances respect the symmetries underlying the differential geometry of quantum states, while non-degeneracy and superposition sensitivity (Axiom~\ref{ax:nondegeneracy}) capture phase-dependent distinctions that are inherently quantum.

\subsection{Relation to Metric Structures on Quantum State Spaces}
The uniqueness of the Fubini--Study metric (Corollary~
\ref{cor:fs_uniqueness}) aligns with classical results in complex projective geometry \cite{kobayashi1996foundations}: it is the unique unitarily invariant Kähler metric up to scale, providing canonical geodesic structure on $\Proj(\Hilbert)$. Likewise, the Bures and trace distances correspond to monotone Riemannian metrics \cite{petz1996monotone}, showing that our axioms recover and organize known classifications of contractive metrics on quantum state spaces under a unified operationally motivated framework.

\subsection{Unification of Existing Distance Measures}
Within the axiomatic framework, previously disparate distance measures acquire a coherent interpretation. 
Bures distance preserves projective and unitary invariance and superposition sensitivity, being a monotone function of $d_{\text{FS}}$. Measurement-induced distances capture operational distinguishability through specific POVMs but may lack superposition sensitivity. Entanglement-sensitive distances satisfy the entanglement-awareness axiom, revealing correlations beyond single-system geometry.\\
This unification demonstrates that the axioms are not arbitrary but provide a natural framework subsuming and clarifying relationships among known quantum distance measures.
\subsection{Comparison with Classical Axiomatic Frameworks}
Classical metric theories—such as abstract metric spaces \cite{hausdorff1918dimension} and probability distances \cite{cover2006elements}—share structural features with our framework, including symmetry, positivity, and the triangle inequality. However, they lack the quantum-specific principles encoded in Axioms~\ref{ax:ray}--\ref{ax:measurement_contextuality}.

In particular, classical metrics depend only on probability magnitudes and are insensitive to relative phase, contrasting with superposition sensitivity (Axiom~\ref{ax:superposition}), which detects differences between coherent superpositions even when classical distributions coincide. Similarly, classical metrics cannot capture non-classical correlations; entanglement awareness (Axiom~\ref{ax:entanglement_awareness}) formalizes sensitivity to global correlations invisible at the subsystem level. Measurement contextuality (Axiom~\ref{ax:measurement_contextuality}) further distinguishes quantum distances, allowing operational definitions to depend on the chosen POVM, unlike classical distance metric functions.

Finally, ray well-definedness (Axiom~\ref{ax:ray}) and unitary invariance (Axiom~\ref{ax:unitary_invariance}) encode physically meaningful symmetries—global phase irrelevance and basis-change invariance—absent in classical metric theory. Together, these distinctions clarify the uniquely quantum features captured by our axiomatic framework, providing a systematic bridge between classical metric theory and quantum information geometry.

\section{Future Directions}\label{sec:open}
\noindent In summary, we have established a rigorous axiomatic foundation for quantum-inspired distances on projective Hilbert spaces. Our framework identifies five key quantum principles, characterizes the uniqueness of the Fubini--Study metric, classifies existing distances according to axiomatic compliance, and provides operational interpretations linking abstract metrics to tasks such as state discrimination and quantum metrology. This unified perspective clarifies the structural consequences of quantum principles, reveals phenomena like entanglement-geometry complementarity, and lays the groundwork for mixed-state generalizations, computational analysis, and new distance measures tailored to quantum applications.\\

\noindent  While comprehensive, the current axioms have limitations. They are restricted to pure states; extending to mixed states requires reformulating superposition sensitivity and projective invariance. Infinite-dimensional Hilbert spaces introduce convergence and completeness challenges. The entanglement axiom could also be refined to quantify the magnitude of entanglement, not just its presence.

\medskip
\noindent Natural directions for future work include the following open problems:

\begin{enumerate}[label=(\roman*)]
\item {Characterization of admissible functions:} Determine all overlap functions $f: [0,1] \to \mathbb{R}_{\ge 0}$ satisfying the triangle inequality and axioms.
\item {Extension to mixed states:} Reformulate axioms for density matrices, replacing projective invariance with unitary invariance and redefining superposition sensitivity.
\item {Entanglement-based distances:} Determine under what conditions the function
\[
d_E(\ket{\psi},\ket{\varphi}) = |E(\ket{\psi}) - E(\ket{\varphi})|,
\]
satisfies all metric axioms, positivity, symmetry, and the triangle inequality, so that it constitutes a valid distance measure.

\item {Computational complexity:} Determine the query complexity for verifying which axioms a candidate distance satisfies.
\item {Unique geodesic characterization:} Explore whether the Fubini--Study metric is the only Riemannian metric satisfying Axioms~\ref{ax:ray} and~\ref{ax:unitary_invariance}, and characterize the full moduli space otherwise.

\end{enumerate}

\section{Acknowledgment}
The author acknowledges the start-up grant (ASE$016~110000~340516$) provided by the College of Science and Engineering (CoSE), Idaho State University.

 \bibliographystyle{ieeetr}
\bibliography{refs}


\end{document}